\documentclass[letterpaper, 10 pt, conference]{ieeeconf}  

\IEEEoverridecommandlockouts                              

\makeatletter
\let\proof\@undefined
\let\endproof\@undefined
\makeatother
\usepackage{amsthm}
\renewenvironment{proof}[1][\proofname]{\par\itshape\noindent\textit{#1.}\enspace}{\par}

\usepackage{amsmath,amssymb,amsfonts,bm, mathtools, dsfont}
\usepackage{xcolor}
\usepackage[nolist,nohyperlinks]{acronym}
\usepackage{siunitx}
\usepackage{tikz}
\usepackage{pgfplots}
\usepackage{pgfplotstable}
\usepackage{comment}
\usepackage{algorithm}
\usepackage{algpseudocode}
\usepackage{cite}
\usepackage{textcomp}
  
\newtheorem{definition}{Definition}
\newtheorem{theorem}{Theorem}
\newtheorem{lemma}{Lemma}

\usepackage{soul}
\setstcolor{red}

\usepackage{hyperref}
\usepackage{siunitx}
\newcommand{\proposed}{$\mu_p$}
\newcommand{\nominal}{$\mu_n$}
\newcommand{\surrogate}{$\mu_s$}
\usepackage{calc}          
\pgfplotsset{compat=1.18}
\usetikzlibrary{calc}      
\usepgfplotslibrary{colormaps}

\title{\LARGE \bf
Importance Sampling for Statistical Certification of Viable Initial Sets 
}

\author{Elizabeth Dietrich$^{1}$, Hanna Krasowski$^{1}$, Vegard Flovik$^{2}$, Murat Arcak$^{1}$
\thanks{$^{1}$E. Dietrich, H. Krasowski, and M. Arcak are with the University of California, Berkeley, USA. Email: {\tt\small \{eadietri, krasowski, arcak\}@berkeley.edu}.}
\thanks{$^{2}$V. Flovik is with Group Research and Development, DNV. Email: {\tt\small vegard.flovik@dnv.com}}%
\thanks{This work was funded in part by the NSF grant CNS-2111688. E. Dietrich was also supported by an NSF Graduate Research Fellowship.}
}

\begin{document}

\begin{acronym}
    \acro{IS}{Importance Sampling}
    \acro{GP}{Gaussian process}
\end{acronym}

\maketitle
\thispagestyle{empty}
\pagestyle{empty}

\begin{abstract}
We study the problem of statistically certifying viable initial sets (VISs)---sets of initial conditions whose trajectories satisfy a given control specification. 
While VISs can be obtained from model-based methods, these methods typically rely on simplified models. 
We propose a simulation-based framework to certify VISs by estimating the probability of specification violations under a high-fidelity or black-box model. 
Since detecting these violations may be challenging due to their scarcity, we propose a sample-efficient framework that leverages importance sampling to target high-risk regions.
We derive an empirical Bernstein inequality for weighted random variables, enabling finite-sample guarantees for importance sampling estimators. We demonstrate the proposed approach on two systems and show improved convergence of the resulting bounds on an adaptive cruise control benchmark.

\end{abstract}

\section{Introduction}
A central task in the control of autonomous systems is characterizing sets of initial conditions from which trajectories satisfy a given specification under appropriate control actions. We refer to these collectively as viable initial sets (VISs), using “viable” broadly to encompass multiple notions of specification satisfaction. 
Examples of VISs include controlled invariant sets (within which trajectories remain) and backward-reachable or reach-avoid sets (states from which a target can be reached, possibly while avoiding unsafe regions). More generally, VISs capture initial conditions that guarantee satisfaction of temporal logic specifications.

Model-based approaches to construct VISs have been extensively studied  \cite{blanchini1999set,prajna2004safety,MitBayTom05,KurzhanskiVaraiya2006,Aubin2009,tabuada,belta,majumdar2017funnel}. For example, symbolic control methods 
\cite{tabuada,belta} automatically construct VISs as the domain of a controller that enforces a temporal logic specification.
Although many model-based approaches also yield control policies that enforce a specification, the resulting VISs are of independent interest, as they decouple viable initial conditions from the specific control approach.
For instance, these sets can be incorporated as state constraints in model predictive control, providing a safety (or specification-satisfying) envelope within which additional objectives, such as passenger comfort in a self-driving car, can be optimized.
However, VISs derived from model-based methods are not necessarily accurate for real-world systems, where the true dynamics are often infeasible to model or proprietary.

To address these limitations, we start from a candidate VIS, possibly derived from a simplified model, and statistically estimate the probability of specification violation. If sufficiently small, this estimate can support certification; otherwise, it signals the need for a more precise VIS computation.
In this work, we propose a framework that leverages importance sampling 
to target critical regions, enabling efficient quantification and certification of rare-event probabilities.
Specifically, our contributions include:
\begin{itemize}
    \item We develop a sample-efficient, data-driven algorithm for characterizing failure sets and computing the failure probability of a candidate VIS.
    \item We propose a Probably Approximately Correct (PAC) estimator for failure probabilities that leverages importance sampling and empirical Bernstein inequalities. 
    \item We provide an empirical comparison against a standard binomial tail inversion.
\end{itemize}

The remainder of this paper is structured as follows. Sec. \ref{sec:prelim} introduces preliminary concepts. Sec.~\ref{sec:probstate} and Sec.~\ref{sec:method} describe the problem formulation, proposed algorithm, and PAC estimator. Finally, Sec. \ref{sec:results} presents an evaluation on both low- and high-dimensional control benchmarks.

\subsection{Related Works}
Simulation-based validation approaches \cite{Kapinski2016SimulationBasedVerification} aim to efficiently explore low-probability, high-risk scenarios in safety-critical systems. 
Two important paradigms are falsification, which targets individual failure trajectories, and verification, which provides formal certificates of correctness. Common falsification frameworks rely on optimization \cite{Deshmukh2017, taliro2011} or reinforcement learning \cite{wang2020}, and usually cover a limited portion of the failure space. To improve coverage, recent work has incorporated verification techniques, such as barrier certificates \cite{murali2024}. Alternatively, probabilistic verification \cite{katoen2016, prism2006} offers quantitative analysis of stochastic system properties. 
However, these approaches remain limited in scalability and expressiveness for real-world systems. In this work, we explicitly characterize failure sets within candidate VISs, enabling an evaluation of failure coverage.

Failure probability estimation \cite{Corso2021Survey} provides an informative assessment of system safety by characterizing the distribution of rare events. 
In particular, \ac{IS} algorithms have received significant attention in applications such as autonomous driving \cite{okelly2018} and aircraft collision avoidance~\cite{kim2016}, due to their ability to produce proposal distributions closely aligned with the failure distribution. Recent work~\cite{delecki2025} has focused on improving \ac{IS} tractability for large state spaces and long horizons.
While standard \ac{IS} estimates are unbiased in expectation, limited simulations make it challenging to quantify confidence in these estimates. 
Therefore, we propose a PAC-style estimator for \ac{IS} that provides finite-sample concentration guarantees and improves convergence to the true failure probability. 

PAC bounds \cite{pacbound} yield two layers of probability, providing a prediction accuracy guarantee and confidence that the randomly drawn samples make this guarantee valid. PAC guarantees are commonly derived using methods from statistical learning theory, such as Vapnik-Chervonenkis dimension \cite{alamo2009}, validation techniques, such as the holdout method \cite{langford2005tutorial, tempobook}, or scenario-based optimization \cite{campi2018introduction}. However, these techniques often rely on restrictive sampling schemes or specific parameterizations. Concentration inequalities, such as Hoeffding's \cite{hoeffding} and Bennett's \cite{bennett} inequalities, also offer distribution-free, finite-sample guarantees under general sampling schemes, including weighted estimates, and scale well to high-dimensional problems. In this work, we extend an empirical Bernstein bound to explicitly handle weighted random variables with reduced sample complexity.

\section{Preliminaries}
\label{sec:prelim}

We denote sets with calligraphic letters, 
and probability distributions with the letter $\mu$. Specifically, we define a proposed \proposed{}, nominal \nominal{}, and surrogate \surrogate{} distribution.
The expected value of a random variable $X$ is denoted as $\mathbb{E}[{X}]$, and probabilities are denoted with $P$.

\subsection{Empirical Bernstein Bound}
Empirical Bernstein bounds \cite{maurer2009empirical} are concentration inequalities that bound the probability of the empirical mean deviating from its true expectation, using a data-driven estimate of variance to produce tighter, variance-sensitive confidence bounds. Bernstein bounds scale asymptotically with the true variance, yielding tighter guarantees than Hoeffding's inequality for hypotheses of small variance. Furthermore, the empirical formulation circumvents the need for prior variance knowledge, as required by Bennett's inequality or the classical Bernstein bound. The empirical Bernstein bound given by \cite{maurer2009empirical} is as follows:

\begin{theorem}[Theorem 4, \cite{maurer2009empirical}]
\label{thm:empiricalbernstein}
    Let ${X},X_1, \dots, X_N$ be i.i.d. random variables with values in $[0, 1]$ and let $\beta > 0$. Then, with probability at least $1-\beta$, 
    \begin{equation}
    \label{eq:empiricalbernstein}
        \mathbb{E}[{X}] - \frac{1}{N}\sum^N_{i=1}X_i \leq \sqrt{\frac{2 \hat{V}_{X} \ln(2/\beta)}{N}} + \frac{7 \ln(2/\beta)}{3(N-1)}, 
    \end{equation}
    where $\hat{V}_{X}$ is the sample variance, 
    \begin{equation}
        \hat{V}_{X} = \frac{1}{N(N-1)}\sum_{1 \leq i \leq j \leq N} (X_i - X_j)^2.
    \end{equation}
\end{theorem}

\subsection{Binomial Tail Bound}
A binomial tail bound evaluates the failure probability of a hypothesis on i.i.d. Bernoulli random variables. For a sample set of size $M$, it calculates the probability of observing at most $k$ failures, which is equivalent to computing a tail bound of a binomial distribution.
In other words, a binomial tail inversion provides a bound on the true error rate such that the probability of $k$ or fewer failures is at least $\beta$ \cite{ langford2005tutorial}.

\begin{definition}[Binomial Tail Inversion]
\label{def:binomial_tail_inversion}
For $k$ violations out of $M$ scenarios and all $\beta \in (0, 1]$:
    \begin{equation}
        \overline{\text{Bin}}(k, M, \beta) = \max_{e}\Bigl\{ e : \text{Bin}\Bigl(k, M, e \Bigr) \geq \beta \Bigr\}, \,\,\text{where}
        \notag
    \end{equation}
    \begin{equation}
        \text{Bin}\Bigl(k, M, e \Bigr) =  \sum_{j=0}^{k} \binom Mj e^j (1-e)^{M-j}.
        \label{eq:binominversion}
    \end{equation}
    \label{def1}
\end{definition}
\subsection{Importance Sampling (IS)}
Importance Sampling (IS)~\cite{owen2013montecarlo, tokdar2010} refers to a class of Monte Carlo methods that approximate mathematical expressions with respect to a target distribution by weighting outcomes drawn from an alternative distribution. 
\ac{IS} approaches intentionally bias a sampling distribution towards regions of interest, increasing the frequency with which rare events are observed. 
In this paper, we focus on the use of \ac{IS} to provide an unbiased estimate of the true failure probability under the nominal distribution \cite{Corso2021Survey}.
The relationship between the proposed \proposed{} and nominal distribution \nominal{} is quantified through the likelihood ratio, or importance weight, 
\begin{equation}
\label{eq:likelihoodratio}
    w(\theta) = \frac{\mu_n(\theta)}{\mu_p(\theta)},
\end{equation}
which acts as a correction factor that compensates for sampling from \proposed{} instead of \nominal{}.
Naive applications of \ac{IS} can result in significant variance inflation due to unbounded importance weights. To mitigate this, defensive \ac{IS} \cite{Hesterberg01051995, Owen01032000} utilizes a mixture proposal to form the proposed distribution \proposed{}. This proposed distribution combines the surrogate \surrogate{} and nominal \nominal{} distributions to ensure sufficient tail coverage in \proposed{} and bounded importance weights. 

\section{Problem Statement}
\label{sec:probstate}
Consider a dynamical system of the form 
\begin{equation}\label{sys:eq}
    \dot{x} = f(x, d), \quad x(0)=\theta,
\end{equation}
where $x$ represents the state and $d$ represents disturbances. Uncertain system parameters can be incorporated into the model by treating them as additional state variables with zero derivatives, so their values are captured in the initial condition vector $\theta$.
We assume that this model incorporates a controller designed to satisfy given specifications and that a candidate VIS $\mathcal{C}$---set of initial conditions $\theta$ which yield trajectories satisfying the specification---is available. Since computing VISs is generally intractable for a high-fidelity model, we assume the candidate VIS is obtained from a simplified model and may contain points $\theta\in \mathcal{C}$ that fail the specification under the true dynamics (\ref{sys:eq}). 
Additionally, we assume the existence of an oracle $\phi$ that evaluates the trajectories of (\ref{sys:eq}) and determines whether a given $\theta\in \mathcal{C}$  leads to a failure.

Given a probability distribution \nominal{} supported on $\mathcal{C}$, our goal is to obtain a bound $\epsilon$ on the failure probability $P_{fail}$ from $N$ i.i.d. samples $\{\theta_i\}_{i=1}^N$ that holds with confidence $1-\beta$, i.e.,
\begin{align}\label{eq:problemstatement_maingoal} 
    &P^N( P_{fail} \leq \epsilon) \geq 1 - \beta,
\end{align}
where $\beta$ is a user-specified parameter.

Failures are likely concentrated in regions of $\mathcal{C}$ that have low probability under the nominal distribution \nominal{}. Therefore, 
we first characterize a
failure-prone set $\mathcal{F} \subset \mathcal{C}$, 
and pursue an \ac{IS} method that targets the critical region $\mathcal{F}$ while appropriately weighting samples to yield an unbiased estimate under the nominal distribution.

\section{Methodology}
\label{sec:method}
When failure events are rare, standard Monte Carlo approaches require many samples to accurately estimate the true failure probability.
To address this challenge, we propose a method for certifying the probability of rare failure events in a candidate VIS $\mathcal{C}$ through adaptive exploration and \ac{IS}. 
In this section, we (i) employ an adaptive surrogate model to efficiently explore $\mathcal{C}$,
(ii) leverage \ac{IS} to provide a high-confidence certification of the true failure probability $P_{fail}$, and (iii) characterize a failure-prone set $\mathcal{F}$ using Gaussian processes (GP) and tractable set representations.
Combining these elements, we provide an algorithm that uses the set $\mathcal{F}$ as the support of a surrogate distribution, which is then used in IS and failure probability estimation.

\subsection{Adaptive Exploration and Importance Sampling}
We let $\theta {\sim} \mu_n$, where $\mu_n$ has support $\mathcal{C}$. Let $\delta(\theta)$ denote the resulting system trajectory and $\ell(\delta(\theta))$ be a binary loss function indicating whether a trajectory results in failure:
    \begin{equation}
        \ell(\delta(\theta)) =  \begin{cases}
                                1 & \text{if } \phi \text{ returns failure}
                                \\
                                0 & \text{otherwise}.
                                \end{cases}
    \end{equation}
We aim to derive a PAC bound on the failure probability, $$P_{fail} = \mathbb{E}_{\mu_n}[\ell(\delta(\theta))].$$
Since failures are rare under the nominal distribution \nominal{}, we sample from a proposed distribution \proposed{}.
The proposed distribution is constructed using a defensive mixture, which ensures the importance weight~\eqref{eq:likelihoodratio} remains strictly bounded,
\begin{equation} \label{eq:proposed_distirbution}
    \mu_p(\theta) = \alpha \mu_s(\theta) + (1-\alpha)\mu_n(\theta),
\end{equation}
where $\alpha \in (0,1)$ and \surrogate{} denotes a surrogate distribution. In particular, \surrogate{} biases the sampling towards rare, but critical, regions associated with failure events, increasing the number of informative samples and reducing the variance of the \ac{IS} estimator.
The construction of the defensive mixture is critical for deriving finite-sample concentration guarantees and ensuring all importance weights are uniformly bounded. 
Under this mixture, the importance weight (\ref{eq:likelihoodratio}) is bounded by
    \begin{equation}
        W_{max} = \frac{1}{1-\alpha},
    \end{equation}
    since
    \begin{equation}\notag
       w(\theta) = \frac{\mu_n(\theta)}{\alpha \mu_s(\theta) + (1-\alpha)\mu_n(\theta)} \leq \frac{\mu_n(\theta)}{(1-\alpha)\mu_n(\theta)} = \frac{1}{1-\alpha}.
    \end{equation}
The parameter $\alpha$ governs reliance on the surrogate model versus the nominal distribution. A large $\alpha$ emphasizes the surrogate's guidance, aggressively sampling in regions deemed important, while a small $\alpha$ preserves nominal coverage, yielding more conservative behavior.

\subsection{Estimator and Certification}
To bound the failure probability $P_{fail}$ using samples drawn from the proposed distribution \proposed{}, we re-weight observed outcomes to account for the sampling bias introduced by the surrogate model. Through this re-weighting procedure, we obtain the importance-weighted loss. 
\begin{definition}[Importance-Weighted Loss]
The importance-weighted loss is a bounded, random variable with values in $[0, W_{max}]$,
\begin{equation}
        Z(\theta) \coloneq w(\theta)\ell(\delta(\theta)), 
    \end{equation}
    whose samples are obtained from those of $\theta$:
    \begin{equation}\label{Zi-def}
        Z_i = w(\theta_i)\ell(\delta(\theta_i)), \quad \forall i = 1, \dots, N.
    \end{equation}
\end{definition}
\noindent By construction, ${Z}(\theta)$ is unbiased---its expectation under \proposed{} recovers the true failure probability under \nominal{}. 
\begin{lemma}[Expected Value of ${Z}$]
\label{lemma:expectedvalue}
    \begin{equation}
    \mathbb{E}_{\theta \sim \mu_p}[{Z}(\theta)] = \mathbb{E}_{\theta \sim \mu_{n}}[\ell(\delta(\theta))] = P_{fail}.
    \end{equation}
    \end{lemma}
\begin{proof}
    \begin{align}
    \mathbb{E}_{\theta \sim \mu_{p}}[{Z}(\theta)] &= \int_\mathcal{C} {Z}(\theta) \mu_{p}(\theta) \hspace{1mm} d\theta \notag \\
        & = \int_\mathcal{C} \frac{\mu_{n}(\theta)}{\mu_{p}(\theta)}\ell(\delta(\theta))\mu_{p}(\theta) \hspace{1mm} d\theta \notag \\
        & = \int_\mathcal{C} \ell(\delta(\theta))\mu_{n}(\theta)\hspace{1mm} d\theta \notag \\
        & = \mathbb{E}_{\theta \sim \mu_{n}}[\ell(\delta(\theta))] = P_{fail}. \notag
    \end{align}
\end{proof}
Since ${Z}$ is a random variable whose expectation recovers the nominal failure probability, we can estimate $P_{fail}$ from a finite set of $N$ i.i.d. samples $\{\theta_i \}^N_{i=1}$ drawn from the proposed distribution \proposed{}. However, the introduction of importance weights results in an estimator that no longer follows a Bernoulli distribution. Therefore, standard binomial confidence bounds (e.g., Eq. \eqref{eq:binominversion}) are not applicable. 

To obtain an upper-bound on $P_{fail}$, we instead rely on the empirical mean of ${Z}$:
    \begin{equation}
        {\overline{Z}} = \frac{1}{N}\sum^{N}_{i=1} w(\theta_i)\ell(\delta(\theta_i)),
    \end{equation}
as well as the empirical variance:
    \begin{equation}
        \hat{V}_{{Z}} = \frac{1}{N-1} \sum^{N}_{i=1}(w(\theta_i)\ell(\delta(\theta_i))-{\overline{Z}})^2.
    \end{equation}
We then derive a concentration inequality tailored to bounded, weighted random variables.
\begin{theorem}[Importance-Weighted PAC Guarantee]
\label{thm:weightedbernstein}
Let $\{Z_i\}_{i=1}^N$ 
be obtained as in (\ref{Zi-def}) from $\{\theta_i \}^N_{i=1}$ drawn from the proposed distribution \proposed{}. 
Given $\beta \in (0,1)$, 
\begin{equation}\label{PAC}
    P^N(P_{fail} \leq \epsilon) \geq 1 - \beta, 
\end{equation}
where
\begin{equation}\label{eq:accuracy_weightedbernstein}
    \epsilon \coloneqq {\overline{Z}} + \sqrt{\frac{2\hat{V}_{{Z}} \ln(2/\beta)}{N}} + \frac{7 \ln (2/\beta)}{3(N-1)}W_{max}.
\end{equation}
\begin{proof}
    Since ${Z}(\theta)= w(\theta)\ell(\delta(\theta))$, where $\ell(\delta(\theta)) \in \{0,1\}$, and $w(\theta) \in [0,W_{max}]$, it follows that 
    $Z(\theta) \in [0,W_{max}]$.
    Define the scaled random variables ${X}= {{Z}(\theta)}/{W_{max}}$, $X_i=Z_i/{W_{max}}$, $i=1,\cdots,N$,  so that ${X},X_1,\cdots,X_N \in [0, 1]$.  Substituting 
    \begin{equation} \notag
        \mathbb{E}[{X}] = \frac{\mathbb{E}[{Z}(\theta)]}{W_{max}}, \quad \frac{1}{N}\sum^N_{i=1}X_i = \frac{{\overline{Z}}}{W_{max}}, \quad \hat{V}_{{X}} = \frac{\hat{V}_{Z}}{W_{max}^2}
    \end{equation}
   in
   Thm. \ref{thm:empiricalbernstein} Eq. \eqref{eq:empiricalbernstein} we obtain
    \begin{equation} \notag
        \mathbb{E}[{Z}] - {\overline{Z}} \leq \sqrt{\frac{2 \hat{V}_{Z} \ln(2/\beta)}{N}} + \frac{7 \ln(2/\beta)}{3(N-1)}W_{max},
    \end{equation}
    with probability at least $1-\beta$. Following from Lemma \ref{lemma:expectedvalue}, $\mathbb{E}[{Z}] = P_{fail}$ and thus
 \begin{equation} \notag
         P_{fail} \leq {\overline{Z}} + \sqrt{\frac{2 \hat{V}_{Z} \ln(2/\beta)}{N}} + \frac{7 \ln(2/\beta)}{3(N-1)}W_{max} =\epsilon.
     \end{equation}
    \qed
\end{proof}
\end{theorem}

Thm. \ref{thm:weightedbernstein} is adaptive to the empirical variance $\hat{V}_{{Z}}$ of the importance-weighted estimator. When the surrogate model effectively concentrates sampling in failure regions, the empirical variance decreases, tightening the bound. The dependence on $W_{max}$ captures the cost of enforcing bounded importance weights. Therefore, in the ideal limit, where the proposal precisely matches the failure region, $\hat{V}_{Z}$ approaches zero and Thm. \ref{thm:weightedbernstein} converges at a rate $\mathcal{O}(N^{-1})$. 
In contrast, standard Monte Carlo certifications are governed by binomial sampling fluctuations, converging at a rate $\mathcal{O}(N^{-1/2})$ \cite{Caflisch_1998}.

\subsection{Failure Set Characterization}
The result of Thm.~\ref{thm:weightedbernstein} is incorporated into Alg.~\ref{alg:solutionsteps}, which relies on two main assumptions: (i) the density functions of \nominal{} and \proposed{} are tractable for computing the importance weight~\eqref{eq:likelihoodratio} and (ii) \nominal{} is supported on the candidate set $\mathcal{C}$. These assumptions can be satisfied by employing set representations for $\mathcal{C}$ that bound \nominal{} and approximating general distributions with common parameterizations that admit analytical probability density functions. 
Under these assumptions, the proposed algorithm computes a failure-prone set $\mathcal{F}$ and its associated distribution \surrogate{}, enabling the use of Thm.~\ref{thm:weightedbernstein} with \proposed{}. To characterize $\mathcal{F}$ with a black-box model, we apply GPs to learn a level set that largely contains the failure region, and approximate this region with a set of polytopes. These polytopes define $\mathcal{F}$, and their volume is used to construct a uniform, tractable \surrogate{}. Given $\mathcal{F}$ and \surrogate{}, the bound $\epsilon$ on $P_{fail}$ is readily computed.

\begin{algorithm}[tb]
\caption{Failure Characterization}
\label{alg:solutionsteps}
\begin{algorithmic}[1]
\Require candidate set $\mathcal{C}$, nominal distribution \nominal{}, failure oracle $\phi$, initial conditions $\theta$, weighting parameter $\alpha$, confidence parameter $\beta$, sample size $N$
\Ensure failure-prone set $\mathcal{F} \subset \mathcal{C}$, bound $\epsilon$ on $P_{fail}$ 
\State $\mathcal{F} \leftarrow \text{compute\_failure\_set}(\mathcal{C}, \phi)$
\State $\mu_s \leftarrow \text{probability\_distribution}(\mathcal{F})$
\State $\mu_p \leftarrow \alpha \mu_s(\theta) + (1 - \alpha) \mu_n(\theta)$ \Comment{according to Eq.~\eqref{eq:proposed_distirbution}}
\State $\epsilon \leftarrow \text{IS\_weighted\_PAC}(N, \mu_p, \mathcal{F})$ \Comment{Eq.~\eqref{eq:accuracy_weightedbernstein}}
\end{algorithmic}
\end{algorithm}

\section{Empirical Results}
\label{sec:results}
We demonstrate Alg.~\ref{alg:solutionsteps} on an Adaptive Cruise Control (ACC) braking scenario and quadrotor flight control.\footnote{For reproducibility, we provide our implementation at: \href{https://github.com/eadietri/ISViableSets_CDC.git}{github.com/eadietri/ISViableSets\_CDC.git}} 
The low-dimensional ACC system admits an analytical estimate of the candidate VIS $\mathcal{C}$, enabling accurate evaluation of our results. In contrast, estimation of $\mathcal{C}$ for the high-dimensional quadrotor system is challenging, highlighting the scalability of the proposed method and bound. 
In this section, we provide details on GP-based $\mathcal{F}$ characterization and report results of Alg. \ref{alg:solutionsteps} for the two systems. 

\begin{figure}[tb]
    \centering
    \resizebox{.9\linewidth}{!}{
    \definecolor{safe}{rgb}{0.19, 0.55, 0.91}
\pgfplotsset{compat=1.18}

\pgfplotstableread[col sep=comma]{figures/data/gp_grid.csv}\gpgrid
\pgfplotstableread[col sep=comma]{figures/data/gp_contour.csv}\gpcontour
\pgfplotstableread[col sep=comma]{figures/data/gp_samples.csv}\gpsamples
\pgfplotstableread[col sep=comma]{figures/data/gp_polytope.csv}\gppolytope
\pgfplotstableread[col sep=comma]{figures/data/acc_viable_polytope.csv}\viableset

\begin{tikzpicture}

\begin{axis}[
    name=mainplot,
    width=0.45\textwidth,
    xlabel={Initial Gap [m]},
    ylabel={Leader Velocity [m/s]},
    xmin=0, xmax=10,
    ymin=0, ymax=10,
    axis background/.style={fill=white},
    colorbar=false,
    legend style={at={(1.05,1)}, anchor=north west}
    grid=major,
]

    \addplot[
    only marks,
    mark=x, mark size=1.5pt, thin,
    black,
    restrict expr to domain={\thisrow{label}}{-0.5:0.4},
    ] table[x=h, y=v] {\gpsamples};

    \addplot[
    thick,
    red,
    fill = red,
    fill opacity = 0.25,
    mark=none,
    ] table[x=h, y=v] {\gppolytope};

    \addplot[
    safe,
    fill = safe,
    fill opacity = 0.1,
    mark=none,
    ] table[x=h, y=v] {\viableset};

\end{axis} 

        \begin{axis}[
            at={(mainplot.north east)},
            xshift=-0.5cm,
            yshift=-0.5cm,
            anchor=north east,
            width=4.55cm,
            height=4.5cm,           
            xmin=0, xmax=1.5,
            ymin=0, ymax=4.5,
            grid=major,
            tick label style={font=\scriptsize},
            title style={font=\scriptsize},
            axis background/.style={fill=white}
        ]
        
            \addplot[
        only marks,
        mark=x, mark size=1.5pt, thin,
        black,
        restrict expr to domain={\thisrow{label}}{-0.5:0.4},
        ] table[x=h, y=v] {\gpsamples};
    
        \addplot[
        thick,
        red,
        fill = red,
        fill opacity = 0.25,
        mark=none,
        ] table[x=h, y=v] {\gppolytope};
    
        \addplot[
        safe,
        fill = safe,
        fill opacity = 0.1,
        mark=none,
        ] table[x=h, y=v] {\viableset};

        \end{axis}

\end{tikzpicture}
\vspace{0.2cm}
    }
    \caption{ACC candidate VIS $\mathcal{C}$ (blue) and learned failure-prone set $\mathcal{F}$ (red) with zoom in of boundary region. 
    \vspace{-2mm}}
\label{fig:gpresults_ACC}
\end{figure}
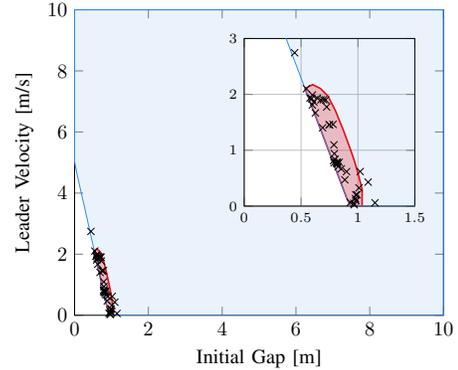

\subsection{Practical Failure Set Construction}
\label{sec:empirialGP}
To efficiently characterize the failure set $\mathcal{F}$, we actively explore $x_0 \in \mathcal{C}$, querying an oracle $\phi$ for failure labels. 
We utilize a least-squares GP classifier, implemented as GP regression on binary observations, and employ a $\gamma$-greedy acquisition strategy to balance boundary refinement with global exploration.
With probability $\gamma$, we uniformly sample at random to mitigate local convergence.  With probability $1-\gamma$, we select the point which minimizes the Straddle score: $S(x) = |(m + m_0) - 0.5| - \kappa \sigma(x)$, where $m$ denotes the GP's predictive mean, $m_0$ is the prior mean, $\kappa$ controls exploration, and $\sigma$ denotes predictive uncertainty. Thus, the strategy prioritizes uncertain points near the decision boundary. 
Given the learned GP level set, we approximate the failure-prone set $\mathcal{F}$. For simplicity, we fit a single polytope given the convex hull of the GP level set, yielding a potentially over-approximative set. 
We then assign the surrogate distribution \surrogate{} to $\mathcal{F}$. Note that when failures are scattered, $\mathcal{F}$ can be tightened by clustering failure regions, fitting polytopes to each cluster, and defining \surrogate{} over their union.
Finally, we construct the defensive mixture by combining \surrogate{} with \nominal{} over $\mathcal{C}$, enabling computation of the PAC bound~\eqref{eq:accuracy_weightedbernstein}. 

\begin{figure}[tb]
    \centering
    \resizebox{.9\linewidth}{!}{\pgfplotstableread[col sep=comma]{figures/data/acc_convergence.csv}\accdata

\pgfplotsset{
    accaxis/.style={
        width=7cm, height=6cm,
        xlabel={$M$},
        grid=both,
        grid style={line width=0.2pt, draw=gray!30},
        major grid style={line width=0.4pt, draw=gray!60},
        tick label style={font=\small},
        label style={font=\small},
        title style={font=\small},
        legend style={
        font=\scriptsize, 
        at={(0.97,0.97)}, 
        anchor=north east,
        legend image post style={scale=0.25}},
        mark size=1.5pt,
    }
}

\begin{tikzpicture}

\begin{axis}[
    accaxis,
    xmode=log, ymode=log,
    name=plot1,
    xmin=1000, xmax=1000000,
    xtick={1000, 10000, 100000, 1000000},
    xticklabels={$10^3$, $10^4$, $10^5$, $10^6$},
    ymin=7e-3, ymax=3e-2,
    ylabel={$\epsilon$},
]
    \addplot[black, mark=none, forget plot]
        table[x=M, y=true_prob] {\accdata};
    \addlegendimage{black}
    \addlegendentry{True Probability}

    \addplot[magenta, thick, dashed]
        table[x=M, y=mc_cp] {\accdata};
    \addlegendentry{Binomial Tail Inversion}

    \addplot[dashed, thick, color=teal]
        table[x=M, y=bound_is_0p15] {\accdata};
    \addlegendentry{IS ($\alpha=0.15$)}

    \addplot[dashed, thick, color=violet]
        table[x=M, y=bound_is_0p35] {\accdata};
    \addlegendentry{IS ($\alpha=0.35$)}
\end{axis}

\end{tikzpicture}}
    
    \vspace{0.5em}
    
    \resizebox{.9\linewidth}{!}{\pgfplotstableread[col sep=comma]{figures/data/acc_convergence.csv}\accdata

\pgfplotsset{
    accaxis/.style={
        width=7cm, height=6cm,
        xlabel={$M$},
        grid=both,
        grid style={line width=0.2pt, draw=gray!30},
        major grid style={line width=0.4pt, draw=gray!60},
        tick label style={font=\small},
        label style={font=\small},
        title style={font=\small},
        legend style={
            font=\scriptsize, 
            at={(0.03,0.03)},       
            anchor=south west,      
            mark size=1.5pt,
            legend image post style={scale=0.25}
        },
    }
}

\begin{tikzpicture}


\begin{axis}[
    accaxis,
    xmode=log, ymode=log,
    name=plot3,
    xmin=1000, xmax=1000000,
    ymin=1e-5, ymax=5e-2,
    xtick={1000, 10000, 100000, 1000000},
    xticklabels={$10^3$, $10^4$, $10^5$, $10^6$},
    ytick={1e-5, 1e-4, 1e-3, 1e-2},
    yticklabels={$10^{-5}$, $10^{-4}$, $10^{-3}$, $10^{-2}$},
    ylabel={$|\text{Bound} - \text{Truth}|$},
]
    \addplot[magenta, thick, dashed]
        table[x=M, y=excess_mc_cp] {\accdata};
    \addlegendentry{Binomial Error}

    \addplot[thick, dashed, color=teal]
        table[x=M, y=excess_is_0p15] {\accdata};
    \addlegendentry{IS Error ($\alpha=0.15$)}

    \addplot[thick, dashed, color=violet]
        table[x=M, y=excess_is_0p35] {\accdata};
    \addlegendentry{IS Error ($\alpha=0.35$)}

    \addplot[red, mark=none]
        table[x=M, y=slope_05] {\accdata};
    \addlegendentry{Slope $-0.5$}

    \addplot[black, mark=none]
        table[x=M, y=slope_10] {\accdata};
    \addlegendentry{Slope $-1.0$}
\end{axis}

\end{tikzpicture}}
    \caption{Top: ACC $\epsilon$ convergence to true failure probability. Bottom: Convergence rate of ACC estimators. The IS estimator converges faster to the true failure probability, closely following a $-1.0$ slope, whereas the binomial estimator converges with more variance at a rate close to $-0.5$. In this case study, $\alpha$ has a negligible effect on the IS estimator.
    }
    \label{fig:accbounds}
\end{figure}
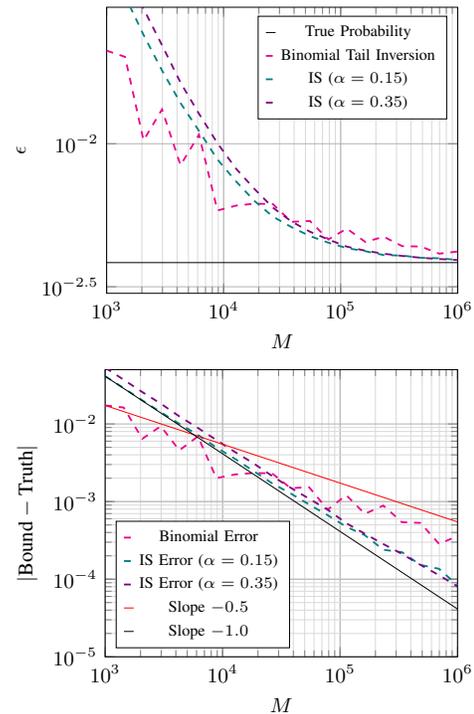

\subsection{Adaptive Cruise Control}
We investigate an ACC scenario using simplified point-mass vehicle models, where a follower vehicle must brake to avoid colliding with a leader vehicle\cite{devonport2020data}. If the distance between the two vehicles becomes zero at any time instance, a collision occurs. 
We take $\mathcal{C}$ as \cite[Eq. (12)]{devonport2020data} assuming optimistic deceleration of $\SI{8}{\meter\per\second\squared}$, representative of typical hard braking, and define a uniform nominal distribution \nominal{} supported on $\mathcal{C}$. We utilize GP regression, as described in Sec. \ref{sec:empirialGP}, where the oracle $\phi$ indicates failure if the two vehicles collide during the testing horizon of $\SI{10}{\second}$.

To simulate inaccuracies in the model used to compute $\mathcal{C}$, we assume more conservative braking dynamics, decreasing the deceleration term to $\SI{3}{\meter\per\second\squared}$, typical for highway driving. As seen in Fig. \ref{fig:gpresults_ACC}, the learned failure set $\mathcal{F}$ occupies a small region of the candidate VIS $\mathcal{C}$ near its boundary. Due to $\mathcal{F}$'s low volume and boundary position, this region is unlikely to be adequately captured by standard Monte Carlo sampling. 

We compare the \ac{IS} estimator proposed in Sec. \ref{sec:method} with a binomial tail inversion based on standard Monte Carlo sampling from \nominal{}, as illustrated in Fig. \ref{fig:accbounds}. The true failure probability is estimated as $0.00845$ using a Monte Carlo point estimate with $10^8$ samples. The \ac{IS} estimator converges quickly to the true failure probability, whereas the binomial tail inversion exhibits significant variance.  
The binomial estimator follows a convergence rate near a slope of $-0.5$, while the \ac{IS} estimator achieves a rate close to $-1.0$, consistent with the analysis presented in Sec. \ref{sec:method}. 
Additionally, we examine the convergence of the \ac{IS} estimator when varying the proposal mixture coefficient $\alpha$. While $\alpha$ primarily affects the bound through $W_{max}$ and $\hat{V}_Z$, it has negligible impact on convergence in this example.

\subsection{Quadrotor Flight Scenario}
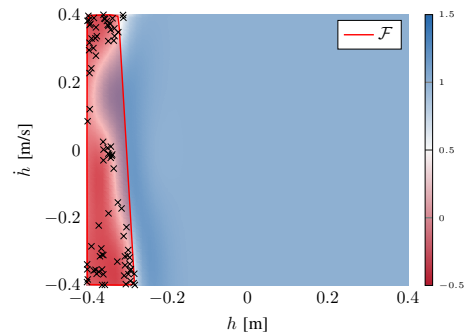
\begin{figure}[tb]
    \centering
    \resizebox{.9\linewidth}{!}{
    \pgfplotsset{compat=1.18}
\pgfplotstableread[col sep=comma]{figures/data/quadrotor_gp_grid.csv}\gpgrid
\pgfplotstableread[col sep=comma]{figures/data/quadrotor_gp_contour.csv}\gpcontour
\pgfplotstableread[col sep=comma]{figures/data/quadrotor_gp_samples.csv}\gpsamples
\pgfplotstableread[col sep=comma]{figures/data/quadrotor_failure_polytope.csv}\gppolytope

\begin{tikzpicture}
\begin{axis}[
    name=plotA,
    width=8cm, height=7cm,
    xlabel={$h$ [m]},
    ylabel={$\dot{h}$ [m/s]},
    xmin=-0.34, xmax=0.34,
    ymin=-0.34, ymax=0.34,
    xtick={-0.34,-0.2,0,0.2,0.34},
    ytick={-0.34,-0.2,0,0.2,0.34},
    colormap={rdbu}{
        rgb255(0pt)=(178,24,43);
        rgb255(250pt)=(247,247,247);
        rgb255(500pt)=(33,102,172)
    },
    point meta min=-0.5,
    point meta max=1.5,
    colorbar,
    colorbar style={
        width=0.15cm,
        font=\tiny,
        xshift=-0.3cm,
    },
]
    \addplot[
        matrix plot*,
        mesh/cols=50,
        point meta=\thisrow{mean},
        shader=interp,
    ] table {\gpgrid};


    \addplot[
        only marks,
        mark=x, mark size=1.5pt, thin,
        black,
        restrict expr to domain={\thisrow{label}}{-0.5:0.4},
    ] table[x=h, y=v] {\gpsamples};

    \addplot[
        thick,
        red,
        fill=red,
        fill opacity=0.25,
        mark=none,
    ] table[x=h, y=v] {\gppolytope};

    \legend{, , $\mathcal{F}$}
\end{axis}
\end{tikzpicture}
    }
    \caption{Quadrotor GP surrogate model with failures and the learned failure set $\mathcal{F}$. Failures are concentrated at low altitudes near the learned boundary. Due to inaccuracy in the GP estimate, $\mathcal{F}$ does not capture all failure points.
    }
    \label{fig:results_quadrotor}
\end{figure}

We consider the 12-dimensional PD-controlled quadrotor system described in \cite[Chap. 8.2]{Meyer2021}.
We aim to verify that the quadrotor remains above $\SI{0.92}{\meter}$ meters after the first second and below a height $h$ of $\SI{1.5}{\meter}$ until the final time horizon of $\SI{5}{\second}$, which can be formalized in signal temporal logic:
\begin{align*}
    \Phi: \; &F_{[0, 1s]} G (h \geq 0.92) \land G(h \leq 1.5).
\end{align*}
The oracle $\phi$ indicates a violation if $\Phi$ is false. 
We define $\mathcal{C}$ as a hyper-rectangle with roll, pitch, and yaw angles and rates in $[-0.05, 0.05]$ and all remaining states in $[-0.34, 0.34]$. 
We define the nominal distribution \nominal{} as a truncated Gaussian with $\sigma = 1$, supported on $\mathcal{C}$. 
Since the specification we aim to verify depends only on height-related components of the quadrotor, the 2-dimensional failure set $\mathcal{F}$ in Fig. \ref{fig:results_quadrotor} is scattered in the remaining dimensions.
Additionally, the nominal distribution's tails result in a low probability of encountering failure, as confirmed in Fig. \ref{fig:mixtureproposal}. 

We estimate bounds on the failure probability $P_{fail}$ using both the \ac{IS} estimator and a binomial tail inversion with $\alpha=0.15$, $\beta=10^{-6}$, $N=10^6$, yielding $\epsilon=0.0138$ and $\epsilon=0.0143$, respectively. For reference, we estimate the true failure probability to be approximately $0.0137$ using a standard Monte Carlo point estimate with $10^8$ samples. Larger values of $\alpha$ improve sampling efficiency when the surrogate accurately captures $\mathcal{F}$, while inaccuracies in the surrogate model diminish the benefits of the \ac{IS} estimator.
In this example, $\alpha=0.15$ provides a balance.

\begin{figure*}[t]
\centering

\pgfplotsset{
  imagepanel/.style={
    width=0.25\textwidth, height=0.25\textwidth,
    xmin=-0.34, xmax=0.34,
    ymin=-0.34, ymax=0.34,
    xtick={-0.34, 0, 0.34},
    ytick={-0.34, 0, 0.34},
    tick align=outside,
    ticklabel style={font=\small},
    axis on top,
  }
}

\pgfplotsset{
  colorbar panel/.style={
    hide axis,
    scale only axis,
    colormap/viridis,
    colorbar,
    height=0.16\textwidth,
    colorbar style={
      yticklabel style={font=\small},
      width=0.2cm,
    },
  }
}
\resizebox{.75\linewidth}{!}{
\begin{tikzpicture}
\pgfplotsset{every axis/.append style={}}
  \begin{axis}[
    imagepanel,
    title={Nominal},
    xlabel={$h$ [m]},
    ylabel={$\dot{h}$ [m/s]},
    ylabel style={font=\small},
    xlabel style={font=\small},
    name=ax1,
  ]
    \addplot graphics [
      xmin=-0.34, xmax=0.34,
      ymin=-0.34, ymax=0.34,
    ] {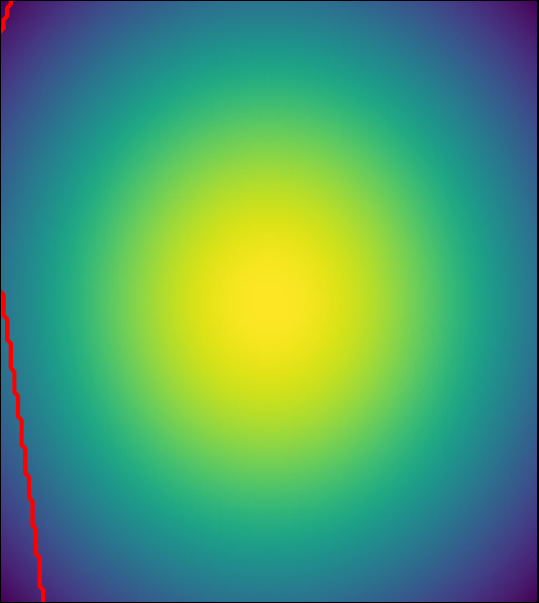};
  \end{axis}

  \begin{axis}[
    imagepanel,
    title={Surrogate},
    xlabel={$h$ [m]},
    yticklabels={},       
    ylabel style={font=\small},
    xlabel style={font=\small},
    name=ax2,
    at={($(ax1.east)+(0.8cm,0)$)},
    anchor=west,
  ]
    \addplot graphics [
      xmin=-0.34, xmax=0.34,
      ymin=-0.34, ymax=0.34,
    ] {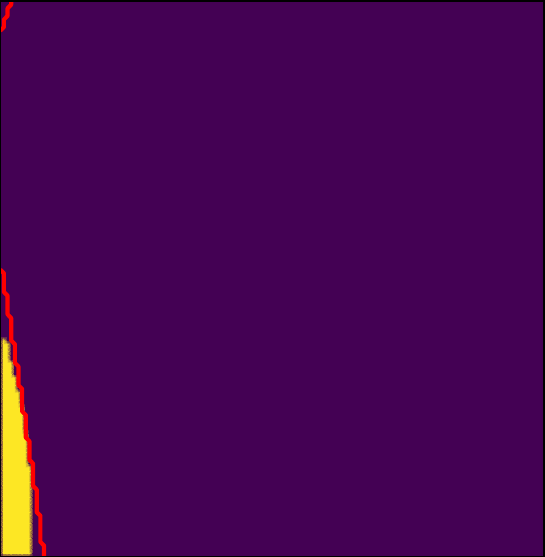};
  \end{axis}

  \begin{axis}[
    imagepanel,
    title={Mixture},
    xlabel={$h$ [m]},
    yticklabels={},
    ylabel style={font=\small},
    xlabel style={font=\small},
    name=ax3,
    at={($(ax2.east)+(0.8cm,0)$)},
    anchor=west,
  ]
    \addplot graphics [
      xmin=-0.34, xmax=0.34,
      ymin=-0.34, ymax=0.34,
    ] {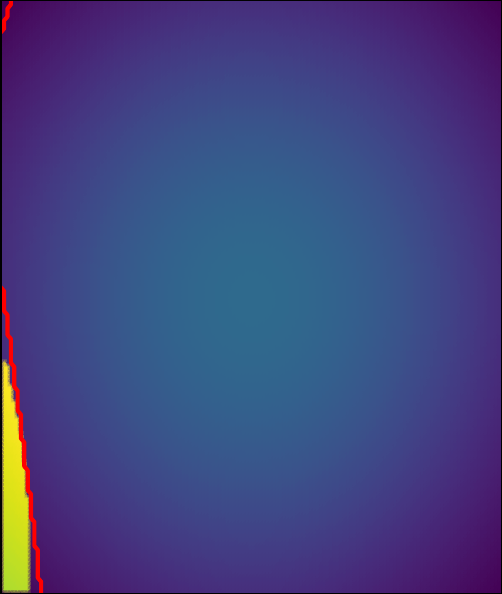};
  \end{axis}

  \begin{axis}[
    colorbar panel,
    at={($(ax3.north east)+(0.2cm,0)$)},
    anchor=north west,
    point meta min=0,
    point meta max=1,
    colorbar style={
      yticklabel style={font=\small},
      width=0.2cm,
      height=0.16\textwidth,
    },
  ]
  \end{axis}

\end{tikzpicture}
}
\vspace*{-0.2cm}
\caption{Nominal, surrogate, and proposed distributions: the defensive mixture increases the probability of sampling from the failure region. 
}
\label{fig:mixtureproposal}
\vspace*{-2.5mm}
\end{figure*}

\section{Conclusion}
This paper presents a data-driven framework for characterizing failure sets and certifying failure probabilities of candidate viable initial sets. By combining importance sampling and adaptive exploration, we efficiently target critical regions that are unlikely to be observed under the nominal distribution. We derive a PAC-style estimator based on an empirical Bernstein inequality that provides finite-sample guarantees for importance-weighted estimates. Our empirical results on both low- and high-dimensional benchmarks demonstrate that the proposed approach improves sample complexity compared to a binomial tail inversion.





\bibliography{literature}
\bibliographystyle{IEEEtran}

\end{document}